\documentclass[11pt]{article}

\usepackage{amsmath, amssymb, amsthm}
\usepackage{epsfig} 

\theoremstyle{plain}
\newtheorem{thm}{Theorem}
\newtheorem{lem}[thm]{Lemma}
\newtheorem{prop}[thm]{Proposition}
\newtheorem{cor}[thm]{Corollary}

\theoremstyle{remark}

\newtheorem{pt}{Part}

\newcommand{\PR}[1]{\ensuremath{{\mathbf{Pr}\left[#1\right]}}}
\newcommand{\EXP}[1]{\ensuremath{{\mathbf{E}\left(#1\right)}}}

\def\ex{{\rm\bf E}}
\def\pr{{\rm\bf Pr}}
\def\var{{\rm\bf Var}}

\newcommand{\remove}[1]{}

\newcommand{\real}{\ensuremath {\mathbb R} }

\newcommand{\UT} {\ensuremath{[0,1)^2}}

\newcommand{\nset} {\ensuremath{\{1,\ldots,n\}}}

\newcommand{\nRG} {\ensuremath{G(\cX;r)}}
\newcommand{\nRGl} {\ensuremath{G(\cX;r_\ell)}}
\newcommand{\nRGu} {\ensuremath{G(\cX;r_u)}}
\newcommand{\nRGR} {\ensuremath{(G(\cX;r))_{r\in\real^+}}}

\newcommand{\cE} {\ensuremath{\mathcal E}}
\newcommand{\cF} {\ensuremath{\mathcal F}}

\newcommand{\cS} {\ensuremath{\mathcal S}}

\newcommand{\cX} {\ensuremath{\mathcal X}}
\newcommand{\cY} {\ensuremath{\mathcal Y}}

\newcommand{\tK} {\ensuremath{\widetilde K}}

\newcommand{\ar}[1] {\mathsf{Area}(#1)}

\title{On the Probability of the Existence of Fixed-Size Components in Random Geometric Graphs\thanks{Partially
supported by the Spanish CYCIT: TIN2007-66523
(FORMALISM). The first and third author are partially supported by 7th Framework under 
contract ICT-2007.82 (FRONTS). The first author was also supported by \emph{La distinci\'{o} per a
la promoci\'{o} de la recerca de la Generalitat de Catalunya, 2002}.}}

\author{J.~D\' \i az$^1$ \qquad D.~Mitsche$^2$ \qquad X.~P\' erez-Gim\'enez$^1$ \smallskip \\
{\small
$^1$Llenguatges i Sistemes Inform\`{a}tics, UPC, 08034
Barcelona }\\
{\small $^2$Institut f\"ur Theoretische Informatik, ETH Z\"urich, 8092 Z\"urich} \\
{\small\tt \{diaz,xperez\}@lsi.upc.edu, dmitsche@inf.ethz.ch}}

\date{}

\begin{document}
\maketitle
\begin{abstract}
In this work we give precise asymptotic expressions on the probability of 
the existence of fixed-size components at the threshold of connectivity
for random geometric graphs.
\end{abstract}
\section{Introduction and basic results on Random Geometric Graphs.}\label{sec:intro}
Recently, quite a bit  of work
has been done on  {\em Random Geometric graphs}, due to the importance of these graphs
as theoretical models for ad hoc networks (for applications we refer to~\cite{Hekmat06}). 
Most of the theoretical results on
random geometric graphs can  be found in the book by 
M.~D.~Penrose~\cite{Penrose03}. In this section we succinctly recall the 
results needed to 
motivate and prove our main theorem.

Given a set of $n$ vertices and a non-negative real $r=r(n)$, each vertex is 
placed at some random position in the unit torus $[0,1)^2$ selected independently
and uniformly at random (u.a.r.). We denote by $X_i=(x_i,y_i)$ the random 
position of vertex $i$ for $i\in\nset$, and let $\cX=\cX(n)=\{X_1,\ldots,X_n\}$.
Note that with probability $1$, no two vertices choose the same position and 
thus we restrict the attention to the case that $|\cX|=n$.
We define $\nRG$ as the random graph having $\cX$ as
the vertex set, and with an edge connecting each pair of vertices $X_i$ and $X_j$ 
in $\cX$ at distance $d(X_i,X_j)\le r$, where $d(\cdot,\cdot)$ denotes the 
Euclidean distance in the torus. 

Unless otherwise stated, all our stated results are asymptotic as $n\to\infty$. We use the following standard notation for the 
asymptotic behaviour of sequences of non-negative numbers $a_n$ and $b_n$: $a=O(b)$, if there exist constants $C$ and $n_0$ such that
$a_n\leq Cb_n$ for $n \geq n_0$. Furthermore, $a=\Omega(b)$ if $b=O(a)$,  $a=\Theta(b)$ if $a=O(b)$ and $a=\Omega(b)$ and finally $a=o(b)$ if
$a_n/b_n \rightarrow 0$ as $n \rightarrow \infty$.  As usual, the
abbreviation a.a.s.\ stands for  {\em asymptotically almost surely}, i.e.\ with 
probability $1-o(1)$. All logarithms in this paper are natural logarithms.

Let $K_1$ be the random variable counting the number of isolated vertices in $\nRG$. By multiplying the probability that one vertex is isolated by the number of vertices 
we obtain
\begin{equation}\label{eq:EK1}
\EXP{K_1} = n (1-\pi r^2)^{n-1} = n e^{-\pi r^2 n - O(r^4 n)}.
\end{equation}
Define $\mu:=n e^{-\pi r^2 n}$. Observe from the previous expression that $\mu$ is closely 
related to $\EXP{K_1}$.
In fact, $\mu=o(1)$ iff $\EXP{K_1=o(1)}$, and if $\mu=\Omega(1)$ then $\EXP{K_1}\sim\mu$.
Moreover, the asymptotic behaviour of 
$\mu$ characterizes the connectivity of $\nRG$. The following proposition is well known: a result similar to item (1) can be found in Corollary 3.1 of~\cite{GuptaKumar} and it can also be found in Section 1.4, p.10 of~\cite{Penrose03}, Item (2) is Theorem 13.11 of~\cite{Penrose03}, and Item (3) can as well be found in Section 1.4, p.10 of~\cite{Penrose03}. For the sake of completeness, we give a simple proof of Proposition~\ref{prop:wellknown} in Section~\ref{sec:corollary}.

%
\begin{prop}\label{prop:wellknown} 
 In terms of $\mu$ the connectivity can be characterized as follows:
 \begin{enumerate}
	\item If $\mu\to0$, then a.a.s.\ $\nRG$ is connected.
	\item If $\mu=\Theta(1)$, then a.a.s.\ $\nRG$ consists of one giant 
component of size $>n/2$ and a Poisson number (with parameter $\mu$) of isolated vertices.
	\item If $\mu\to\infty$, then a.a.s.\ $\nRG$ is disconnected. 
	\end{enumerate}
\end{prop}

{From} the definition of $\mu$ we have that $\mu=\Theta(1)$ iff $r=\sqrt{\frac{\log n \pm O(1)}{\pi n}}$. Therefore we conclude that
the property of connectivity of $\nRG$ exhibits a sharp threshold at 
$r=\sqrt{\frac{\log n}{\pi n}}$.
Note that the previous classification of the connectivity of $\nRG$, 
indicates that if $\mu=\Theta(1)$, the components of size $1$  are predominant and
 those components have the main contribution to the connectivity of $\nRG$. In fact 
 if $\mu=\Theta(1)$, the probability that $\nRG$ has some component of 
size greater than $1$ other than the giant component is $o(1)$. 

On the other hand, M.D. Penrose~\cite{Penrose03} studied the number of components in $\nRG$ that are isomorphic to a given fixed graph; equivalently, he studied the probability of finding 
{\em components} of a given size in $\nRG$. However the range of radii 
$r$ covered by Penrose does not exceed 
the {\em thermodynamical threshold} $\Theta(\sqrt{1/n})$ where a giant component appears 
at $\nRG$, which is  below the connectivity threshold treated in the present paper. 
In fact, a percolation argument in~\cite{Penrose03} only shows that with probability $1-o(1)$ 
no components other than isolated vertices and the giant one exist at the connectivity 
threshold, whithout giving accurate bounds on this probability (see Section 1.4 of~\cite{Penrose03} and Proposition 13.12 and 
Proposition 13.13 of~\cite{Penrose03}).

Throughout the paper we shall consider $\nRG$ with $r=\sqrt{\frac{\log n\pm O(1)}{\pi n}}$. We prove that for such a choice of $r$, 
given a fixed $\ell > 1$, the probability of having components of size exactly $\ell$
is $\Theta\left(\frac1{\log^{i-1} n}\right)$.
 Moreover, in the process of the proof we characterize the different types of components that could exist for such a value of $r$.

\section{Basic definitions and statements of results}

Given a component $\Gamma$ of $\nRG$, $\Gamma$ is {\em embeddable} \remove{if it is contained 
in some square with sides parallel to the axes of the torus and length $1-2r$.
In other words, $\Gamma$ is embeddable}
if it can be mapped into the square $[r,1-r]^2$ by a translation in the torus. 
Embeddable components do not wrap around the torus. 
\remove{Throughout the chapter and often without explicitly mentioning it, we assume in
all geometrical descriptions involving an embeddable component $\Gamma$ that $\Gamma$ is contained in $[r,1-r]^2$ and regard the torus $[0,1)^2$ as the unit square and $d(\cdot,\cdot)$ as the usual Euclidean distance.
Hence terms as ``left'', ``right'', ``above'' and ``below'' are globally defined.
On the other hand,}

\remove{\begin{figure}
\centerline{\epsfig{figure=isolated.eps,height=7cm,width=6cm}}
\caption{Embeddable components on the unit torus.}
\label{embeddable}
\end{figure}}
Components which are not embeddable must have a large size (at least $\Omega(1/r)$).
Sometimes several non-embeddable components can coexist together (see Figure~\ref{non-embeddable}). However, 
there are some non-embeddable components which are so spread around the torus, that they do not allow any 
room for other non-embeddable ones. Call these components {\em solitary}. Clearly, we can have at most one solitary component. We cannot disprove the existence of a solitary component, 
since with probability $1-o(1)$ there exists a giant component of this nature (see Corollary 2.1 of~\cite{GuptaKumar}, implicitly it is also in Theorem 13.11 of~\cite{Penrose03}). 
\begin{figure}
\centerline{\epsfig{figure=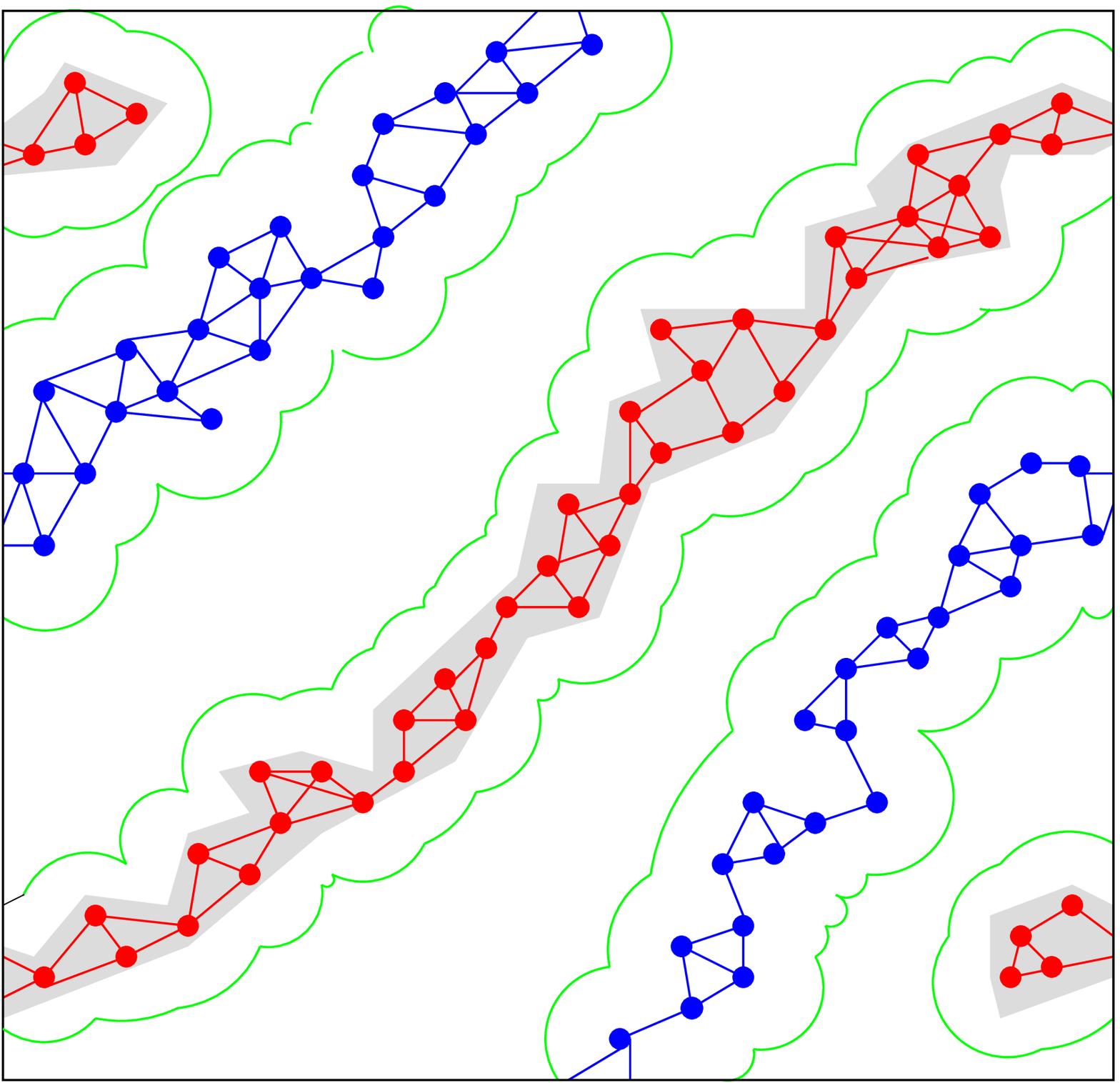,height=4cm,width=4cm}\medspace\medspace\medspace
\medspace\medspace\medspace\epsfig{figure=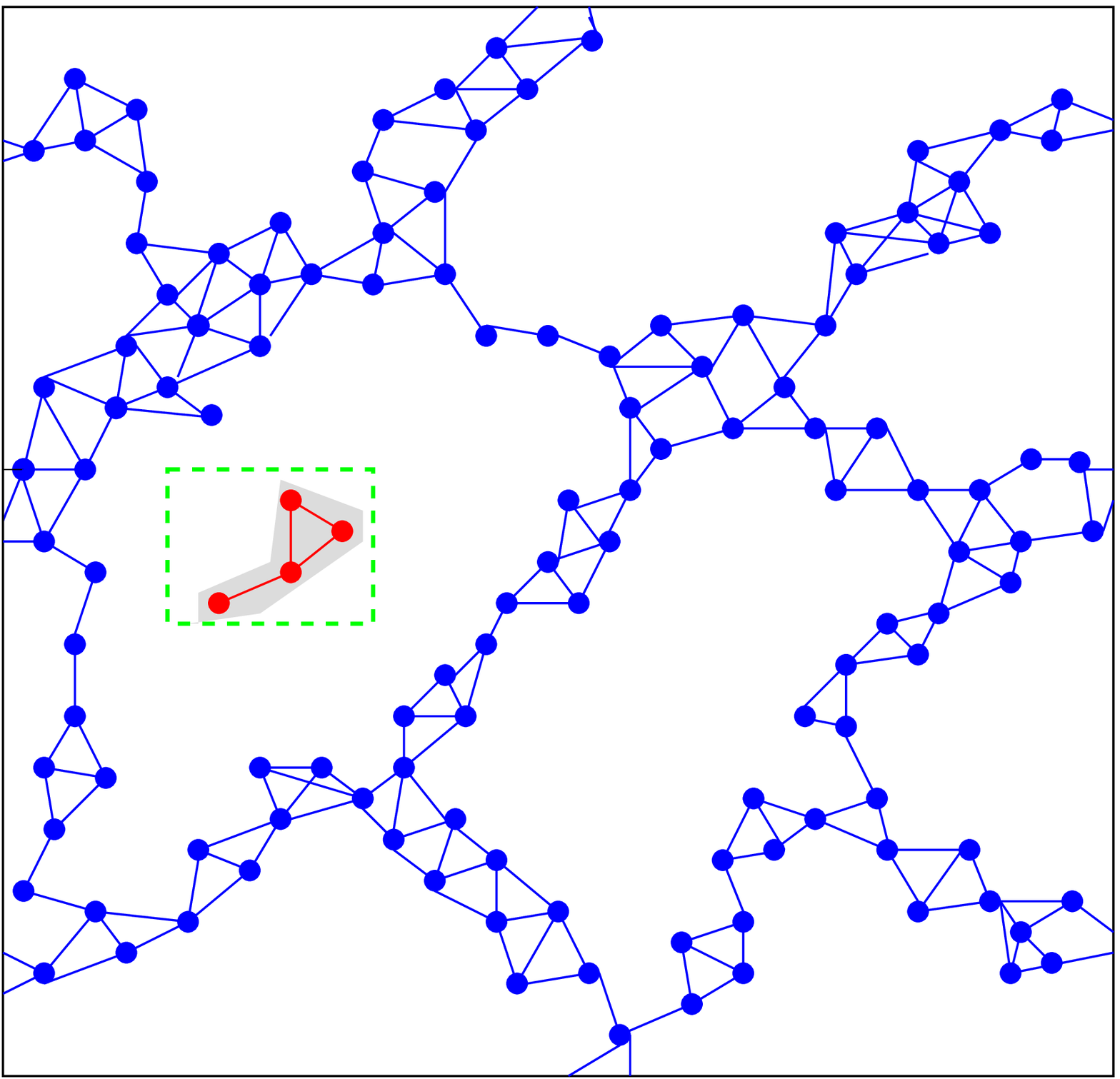,height=4cm,width=4cm}}
\caption{Non-embeddable components on the unit torus. To the left two non-embeddable and non-solitary components,
to the right a solitary non-embeddable and an embeddable component.}
\label{non-embeddable}
\end{figure}
For components which are not solitary, we 
give asymptotic bounds on the probability of their existence according to their size.

Given a fixed integer $\ell \ge 1$, let $K_\ell$ be the 
number of components in $\nRG$ of size exactly $\ell$. For large enough $n$, 
we can assume these to be embeddable, since $r=o(1)$. Moreover, for any fixed $\epsilon>0$, 
let $K'_{\epsilon,\ell}$ be the number of components of size exactly $\ell$, 
which have all their vertices at distance at most $\epsilon r$ from their leftmost one.
Finally, $\tK_\ell$ denotes the number of components of size at least $\ell$ 
and which are not solitary. In Figure~\ref{figstatic} an example of 
a component $\Gamma$ of size exactly $\ell =9$ is given, 
which has all its vertices at distance at most $\epsilon r$ from the leftmost one $u$.

\begin{figure}
\centerline{\epsfig{figure=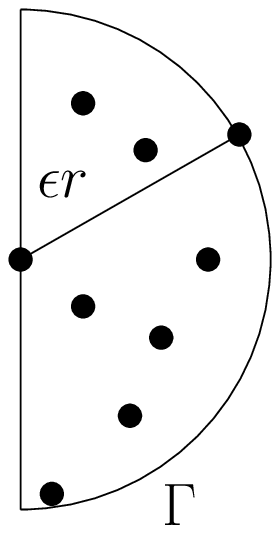,height=4cm,width=2cm}}
\caption{A component $\Gamma$ belonging to $K'_{\epsilon,9}$}
\label{figstatic}
\end{figure}

Notice that $K'_{\epsilon,\ell}\le K_\ell\le\tK_\ell$. However, 
in the following we show that asymptotically all the weight 
in the probability that $\tK_\ell>0$ comes from components which also 
contribute to $K'_{\epsilon,\ell}$ for $\epsilon$ arbitrarily small. 
This means that the more common components of size at least $\ell$ are 
cliques of size exactly $\ell$ with all their vertices close together.

We now have all definitions to state our main theorem, which is proved in Section~\ref{sec:proofmain}.
\begin{thm}\label{thm:static2}
Let  $\ell\ge2$ be a fixed integer.  Let $0<\epsilon<1/2$ be fixed. Assume that $\mu=\Theta(1)$. Then
\[
\PR{\tK_\ell>0} \sim \PR{K_\ell>0} \sim \PR{K'_{\epsilon,\ell}>0}
= \Theta\left(\frac1{\log^{\ell-1} n}\right).
\]
\end{thm}
Given a random set $\cX$ of $n$ points in \UT, let $\nRGR$ be the continuous random graph process describing the evolution of $\nRG$ for $r$ between $0$ and $+\infty$ ($\cX$ remains fixed for the whole process).
Observe that the graph process starts at $r=0$ with all $n$ vertices being isolated, then edges are progressively added, and finally at $r\ge\sqrt2/2$ we have the complete graph on $n$ vertices.
In this context, consider the random variables $r_c=r_c(n)=\inf \{r\in\real^+: \nRG \mbox{ is connected}\}$ and $r_i=r_i(n)=\inf\{r\in\real^+: \nRG \mbox { has no isolated vertex}\}$. 

As a corollary of Theorem~\ref{thm:static2} we obtain an alternative proof of the following well known result (see Theorem 1 of~\cite{Penrose97}):
intuitively speaking, we show that a.a.s.\ $\nRGR$ becomes connected exactly at the same moment when the last isolated vertex disappears. Note that this is stronger than the results stated in the introduction, which just say that the properties of connectivity and having no isolated vertex have a sharp threshold with the same asymptotic characterization (see Proposition~\ref{prop:wellknown}).
\remove{OLD:
 As a corollary of Theorem~\ref{thm:static2} we obtain an alternative proof of the following well known result (see Theorem 1 of~\cite{Penrose97}):
intuitively speaking, we (re)prove that a.a.s.\ exactly at the moment when the last isolated vertex disappears $\nRG$ gets connected. Note that this is stronger than the results stated in the introduction, which just say that the properties of connectivity and having no isolated vertex have the same sharp threshold.
Denote by $r_c=r_c(n)=\inf \{r: \nRG \mbox{ is connected}\}$ and by 
$r_i=r_i(n)=\inf\{r: \nRG \mbox { has no isolated vertex}\}$. 
}
\begin{cor}\label{cor:hitting}
With probability $1-o(1)$, we have $r_c=r_i$.
\end{cor}
The proof of Corollary~\ref{cor:hitting} is given in Section~\ref{sec:corollary}.

\section{Proof of Theorem~\ref{thm:static2}}\label{sec:proofmain}
We state and prove three lemmata from which Theorem~\ref{thm:static2} will follow easily.
\begin{lem}\label{lem:EZei}
Let  $\ell\ge2$ be a fixed integer, and $0<\epsilon<1/2$ be also fixed. 
Assume that $\mu=\Theta(1)$. Then,
\[
\EXP{K'_{\epsilon,\ell}}  = \Theta(1/\log^{\ell-1} n).
\]
\end{lem}
\begin{proof}
First observe that with probability $1$, for each component $\Gamma$ which contributes to $K'_{\epsilon,\ell}$, $\Gamma$ has a unique leftmost vertex $X_i$ and the vertex $X_j$ in $\Gamma$ at greatest distance from $X_i$ is also unique. Hence, we can restrict our attention to this case.

Fix an arbitrary set of indices $J\subset\nset$ of size $|J|=\ell$, with two distinguished elements $i$ and $j$.
Denote by $\cY = \bigcup_{k\in J} X_k$ the set of random points in $\cX$ with indices in $J$.
Let $\cE$ be the following event: 
All vertices in $\cY$ are at distance at most $\epsilon r$ from $X_i$ and to the right of $X_i$; vertex $X_j$ is the one in $\cY$ with greatest distance from $X_i$; and the vertices of $\cY$ form a component $\Gamma$ of $\nRG$.
If $\pr(\cE)$ is multiplied by the number of possible choices of $i$, $j$ and the remaining $\ell-2$ elements of $J$, we get
\begin{equation}\label{eq:EZei}
\ex K'_{\epsilon,\ell} = n(n-1)\binom{n-2}{\ell-2} \pr(\cE).
\end{equation}

In order to bound the probability of $\cE$ we need some definitions.
Let $\rho=d(X_i,X_j)$ and let $\cS$ be the set of all points in the torus $[0,1)^2$ 
which are at distance at most $r$ from some vertex in $\cY$ (see Figure~\ref{figstatic1}).
Notice that $\rho$ and $\cS$ depend on the set of random points $\cY$.

\begin{figure}
\centerline{\epsfig{figure=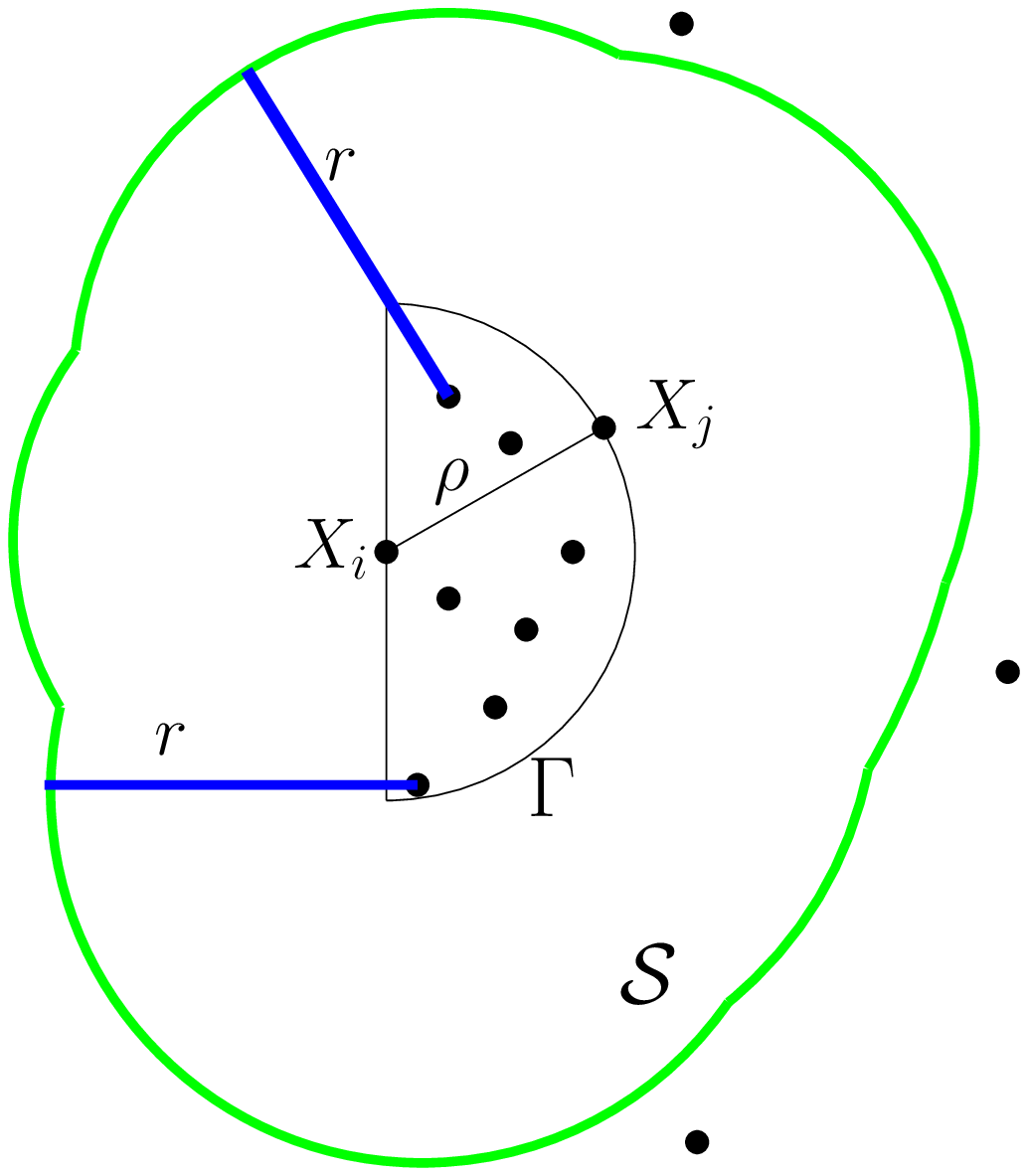,height=4cm,width=4cm}}
\caption{The set $\cS$ for the component $\Gamma$ of Figure~\ref{figstatic}}
\label{figstatic1}
\end{figure}

We first need bounds of $\ar{\cS}$ in terms of $\rho$.
Observe that $\cS$ is contained in the circle of radius $r+\rho$ and center $X_i$, and thus
\begin{equation}\label{eq:Subound}
\ar{\cS}\le\pi(r+\rho)^2.
\end{equation}
Let $i_{\mathsf L}=i$, $i_{\mathsf R}$, $i_{\mathsf T}$ and $i_{\mathsf B}$ be respectively the indices of the leftmost, rightmost, topmost and  bottommost vertices in $\cY$ (some of these indices possibly equal).
Assume w.l.o.g. that the vertical length of $\cY$ (i.e.\ the vertical distance between $X_{i_{\mathsf T}}$ and $X_{i_{\mathsf B}}$) is at least $\rho/\sqrt2$. Otherwise, the horizontal length of $\cY$ has this property and we can rotate the descriptions in the argument. The upper halfcircle with center $X_{i_{\mathsf T}}$ and the lower halfcircle with center $X_{i_{\mathsf B}}$ are disjoint and are contained in $\cS$. 
If $X_{i_{\mathsf R}}$ is at greater vertical distance from $X_{i_{\mathsf T}}$ than from $X_{i_{\mathsf B}}$, then consider the rectangle of height $\rho/(2\sqrt2)$ and width $r-\rho/(2\sqrt2)$ with one corner on $X_{i_{\mathsf R}}$ and above and to the right of $X_{i_{\mathsf R}}$. Otherwise, consider the same rectangle below and to the right of $X_{i_{\mathsf R}}$.
This rectangle is also contained in $\cS$ and its interior does not intersect the previously described halfcircles. Analogously, we can find another rectangle of height $\rho/(2\sqrt2)$ and width $r-\rho/(2\sqrt2)$ to the left of $X_{i_{\mathsf L}}$ and either above or below $X_{i_{\mathsf L}}$ with the same properties. Hence,
\begin{equation}\label{eq:Slbound}
\ar{\cS}\ge \pi r^2 + 2\left(\frac\rho{2\sqrt2}\right)\left(r-\frac\rho{2\sqrt2}\right).
\end{equation}
{From}~\eqref{eq:Subound},~\eqref{eq:Slbound} and the fact that $\rho<r/2$, we can write
\begin{equation}\label{eq:Sbound}
\pi r^2 \left(1 + \frac{1}{6} \frac\rho{r} \right) < \ar{\cS} < \pi r^2 \left(1+\frac52\frac\rho r\right) < \frac{9\pi}{4} r^2.
\end{equation}
Now consider the probability $P$ that the $n-\ell$ vertices not in $\cY$ lie outside $\cS$. Clearly $P=(1-\ar{\cS})^{n-\ell}$. Moreover, by~\eqref{eq:Sbound} and using the fact that $e^{-x-x^2} \le 1-x \le e^{-x}$ for all $x\in[0,1/2]$, we obtain
\[
e^{-(1+5\rho/(2 r))\pi r^2n - (9\pi r^2/4)^2n} < P <   \frac{e^{-(1+\rho/(6r))\pi r^2n}}{(1-9\pi r^2/4)^\ell},
\]
and after plugging in the definition of $\mu$ (recall that $\mu=ne^{-r^2 \pi n}$) we have
\begin{equation} \label{eq:Pbound}
\left(\frac\mu n\right)^{1+5\rho/(2 r)} e^{-(9\pi r^2/4)^2n} < P <  \left(\frac\mu n\right)^{1+\rho/(6r)}    \frac{1}{(1-9\pi r^2/4)^\ell}.
\end{equation}
\remove{\begin{figure}
\centerline{\epsfig{figure=figstatic3.eps,height=4cm,width=4cm}\medspace\medspace\medspace\medspace \epsfig{figure=figstatic2.eps,height=4cm,width=4cm}}
\caption{The upper and lower bounds to $\ar{\cS}$ in the proof of Lemma~1.}
\label{ub}
\end{figure}}

Event $\cE$ can also be described as follows: There is some non-negative real $\rho\le\epsilon r$ such that $X_j$ is placed at distance $\rho$ from $X_i$ and to the right of $X_i$; all the remaining vertices in $\cY$ are inside the halfcircle of center $X_i$ and radius $\rho$;
and the $n-\ell$ vertices not in $\cY$ lie outside $\cS$.
Hence, $\pr(\cE)$ can be  bounded from above (below) by integrating with respect to $\rho$ the probability density function of $d(X_i,X_j)$ times the probability that the remaining $\ell-2$ selected vertices lie inside the right halfcircle of center $X_i$ and radius $\rho$ times the upper (lower) bound on $P$ we obtained in~\eqref{eq:Pbound}:
\begin{equation}\label{eq:PE}
\Theta(1)\, I(5/2)\le \pr(\cE) \le \Theta(1)\, I(1/6),
\end{equation}
where
\begin{align}
I(\beta) &= \int_0^{\epsilon r} \pi\rho \left(\frac\pi2\rho^2\right)^{\ell-2} \frac{1}{n^{1 + \beta\rho/r}}  \,d\rho
\notag\\
&=  \frac2n \left(\frac\pi2 r^2\right)^{\ell-1}  \int_0^{\epsilon} x^{2\ell-3} n^{-\beta x}  \,dx
\label{eq:Ibeta}
\end{align}
Since $\ell$ is fixed, for $\beta=5/2$ or $\beta=1/6$,
\begin{align}
I(\beta) &= \Theta\left(\frac{\log^{\ell-1}n}{n^{\ell}}\right)  \int_0^{\epsilon} x^{2\ell-3} n^{-\beta x}  \,dx
\notag\\
&= \Theta\left(\frac{\log^{\ell-1}n}{n^{\ell}}\right)  \frac{(2\ell-3)!}{(\beta\log n)^{2\ell-2}}
\notag\\
&= \Theta\left(\frac{1}{n^{\ell}\log^{\ell-1}n}\right)
\label{eq:Ibeta2}.
\end{align}
The statement follows from~\eqref{eq:EZei}, \eqref{eq:PE} and~\eqref{eq:Ibeta2}.
\end{proof}

\begin{lem}\label{lem:PY}
Let  $\ell\ge2$ be a fixed integer. Let $\epsilon>0$ be also fixed. Assume that $\mu=\Theta(1)$. 
Then
\[
\PR{\tK_\ell-K'_{\epsilon,\ell}>0} = O(1/\log^{\ell} n).
\]
\end{lem}
\begin{proof}
We assume throughout this proof that $\epsilon\le10^{-18}$, and prove the claim for this case.
The case $\epsilon>10^{-18}$ follows from the fact that
$(\tK_\ell-K'_{\epsilon,\ell}) \le (\tK_\ell-K'_{10^{-18},\ell})$.

Consider all the possible components in $\nRG$ which are not solitary. Remove from these components the ones of size at most $\ell$ and diameter at most $\epsilon r$, and denote by $M$ the number of remaining components. By construction $\tK_\ell-K'_{\epsilon,\ell} \le M$, and therefore it is sufficient to prove that $\pr(M>0)=O(1/\log^{\ell} n)$.
The components counted by $M$ are classified into several types according to their size and diameter. We deal with each type separately.
\setcounter{pt}{0}
\begin{pt}\label{p:small}
Consider all the possible components in $\nRG$ which have diameter at most $\epsilon r$ and size between $\ell+1$ and  $\log n/37$. Call them components of type~\ref{p:small}, and let $M_{\ref{p:small}}$ denote their number.

For each $k$, $\ell+1\le k\le\log n/37$, let $E_k$ be the expected number of components of type~\ref{p:small} and size $k$.
We observe that these components have all of their vertices at distance at most $\epsilon r$ from the leftmost one. Therefore, we can apply the
same argument we used for bounding $\ex K'_{\epsilon,\ell}$ in the proof of Lemma~\ref{lem:EZei}. Note that~\eqref{eq:EZei}, \eqref{eq:PE} and~\eqref{eq:Ibeta} are also valid for sizes not fixed but depending on $n$. Thus, we obtain
\[
E_k \le O(1) n(n-1)\binom{n-2}{k-2} I(1/6),
\]
where $I(1/6)$ is defined in~\eqref{eq:Ibeta}. We use the fact that $\binom{n-2}{k-2}\le(\frac{ne}{k-2})^{k-2}$ and get
\begin{equation}\label{eq:El}
E_k = O(1)  \log n \left(\frac{e}2  \frac{\log n}{k-2}\right)^{k-2} \int_0^{\epsilon} x^{2k-3} n^{-x/6}  \,dx.
\end{equation}
The expression $x^{2k-3} n^{-x/6}$ can be maximized for $x\in\real^+$ by elementary techniques, and we deduce that 
\[
x^{2k-3} n^{-x/6} \le \left(\frac{2k-3}{(e/6)\log n}\right)^{2k-3}.
\]
We can bound the integral in~\eqref{eq:El} and get
\begin{align*}
E_k &= O(1)  \log n \left(\frac{e}2  \frac{\log n}{k-2}\right)^{k-2}  \epsilon \left(\frac{2k-3}{(e/6)\log n}\right)^{2k-3}
\\
&= O(1)   \left(\frac{36}{2e}  \frac{(2k-3)^2}{(k-2) \log n}\right)^{k-2} k.
\end{align*}
Note that for $k\le\log n/37$ the expression $k \left(\frac{36}{2e}  \frac{(2k-3)^2}{(k-2) \log n}\right)^{k-2}$ is decreasing with $k$. Hence we can write
\[
E_k = O \left(\frac{1}{\log^{\ell+1} n}\right), \qquad \forall k \;:\; \ell+3\le k\le\frac1{37}\log n.
\]
Moreover the bounds
$E_{\ell+1}=O(1/\log^{\ell}n)$ and $E_{\ell+2}=O(1/\log^{\ell+1}n)$ are obtained from Lemma~\ref{lem:EZei}, and hence
{\small
\[
\ex M_{\ref{p:small}} = \sum_{k=\ell+1}^{\frac1{37}\log n} E_k
= O\left(\frac{1}{\log^{\ell} n}\right) +  O\left(\frac{1}{\log^{\ell+1} n}\right) + \frac{\log n}{37} O\left(\frac{1}{\log^{\ell+1} n}\right)
=  O\left(\frac{1}{\log^{\ell} n}\right),
\]
}
and then $\pr(M_{\ref{p:small}}>0)\le\ex M_{\ref{p:small}} = O(1/\log^{\ell} n).$
\end{pt}
\begin{pt}\label{p:dense}
Consider all the possible components in $\nRG$ which have diameter at most $\epsilon r$ and size greater than $\log n/37$. Call them components of type~\ref{p:dense}, and let $M_{\ref{p:dense}}$ denote their number.

We tessellate the torus with square cells of side $y=\lfloor(\epsilon r)^{-1}\rfloor^{-1}$ ($y\ge\epsilon r$ but also $y\sim\epsilon r$).
We define a box to be a square of side $2y$ consisting of the union of $4$ cells of the tessellation. Consider the set of all possible boxes.
Note that any component  of  type~\ref{p:dense} must be fully contained in some box (see Figure~\ref{tas0}).

\begin{figure}
\centerline{\epsfig{figure=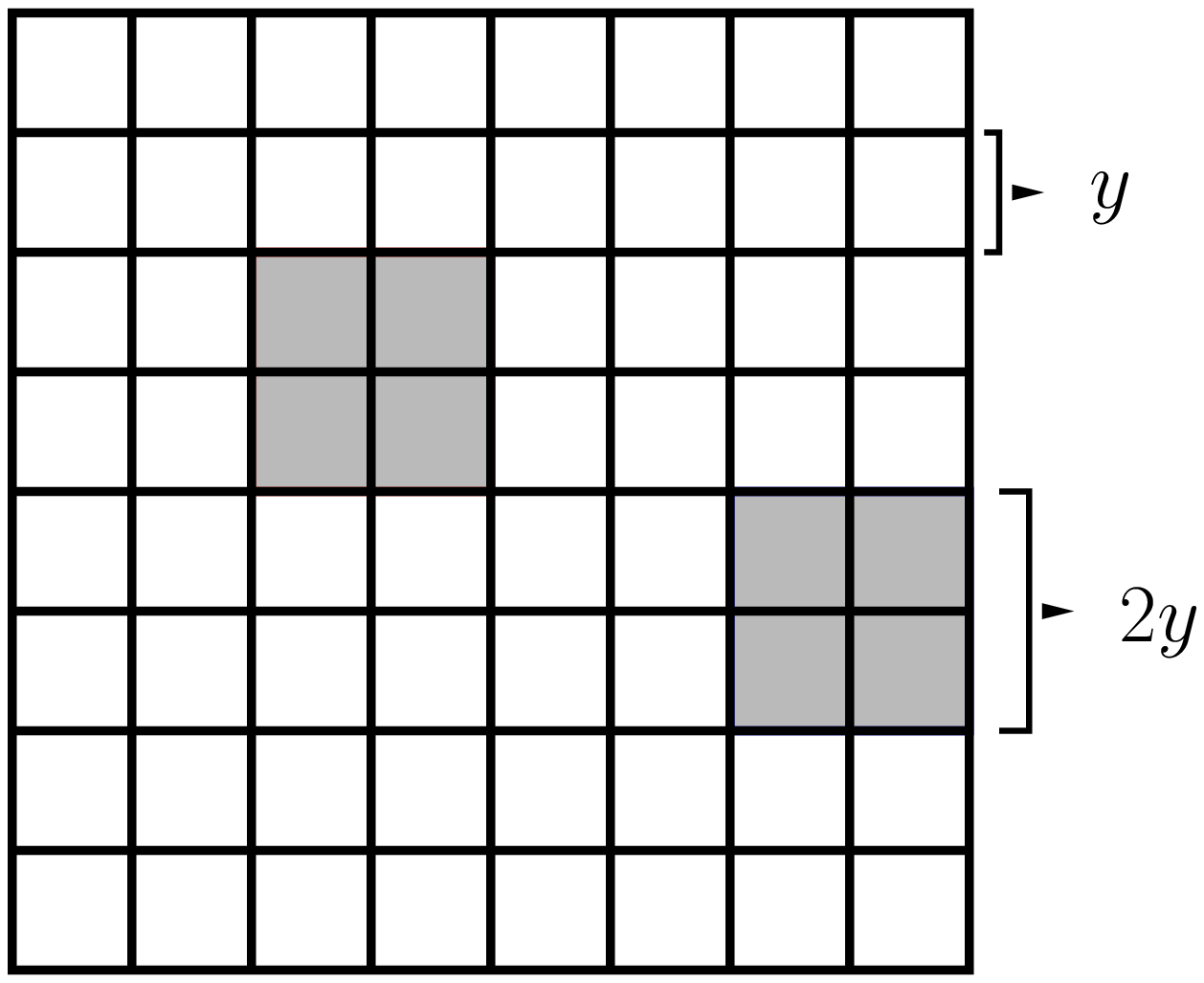,height=3cm,width=4cm}}
\caption{The tessellation for counting components of type~2 with two particular boxes marked.}
\label{tas0}
\end{figure}

Let us fix a box $b$. Let $W$ be the number of vertices which are contained inside $b$. Notice that $W$ has a binomial distribution with mean
$\ex W = (2y)^2n \sim (2\epsilon)^2\log n/\pi$. By
setting $\delta=\frac{\log n}{37\ex W} - 1$  and applying the Chernoff inequality to $W$ (see e.g.~\cite{DiazPetitSerna}, Theorem 12.7), we have
\[
\pr(W>\frac1{37}\log n) = \pr(W>(1+\delta)\ex W) \le \left(\frac{e^\delta}{(1+\delta)^{1+\delta}}\right)^{\ex W}
= n^{ -\frac{(\log(1+\delta) -  \frac\delta{1+\delta})}{37} }.
\]
Note that $\delta \sim \frac{\pi}{148\epsilon^2}-1> e^{79}$, therefore
\[
\pr(W>\frac1{37}\log n)  < n^{-2.1}.
\]
Taking a union bound over the set of all $\Theta(r^{-1})=\Theta(n/\log n)$ boxes, the probability that there is some box with more than $\frac1{37}\log n$ vertices is $O(1/(n^{1.1}\log n))$. Since each component of type~\ref{p:dense} is contained in some box, we have
\[
\pr(M_{\ref{p:dense}}>0)=O(1/(n^{1.1}\log n)).
\]
\end{pt}
\begin{pt}\label{p:large}
Consider all the possible components in $\nRG$ which are embeddable and have diameter at least $\epsilon r$. 
Call them components of type~\ref{p:large}, and let $M_{\ref{p:large}}$ denote their number.

We tessellate the torus into square cells of side $\alpha r$, for some $\alpha=\alpha(\epsilon)>0$ fixed but 
sufficiently small. 
Let $\Gamma$ be a component of type~\ref{p:large}.
Let $\cS=\cS_\Gamma$ be the set of all points in the torus $[0,1)^2$ which are at distance at most $r$ from some vertex in $\Gamma$.
Remove from $\cS$ the vertices of $\Gamma$ and the edges (represented by straight line segments) and denote 
by $\cS'$ the outer connected topologic component of the remaining set. 
By construction, $\cS'$ must contain no vertex in $\cX$ (see Figure~\ref{tas1}, left picture).

\begin{figure}
\centerline{\epsfig{figure=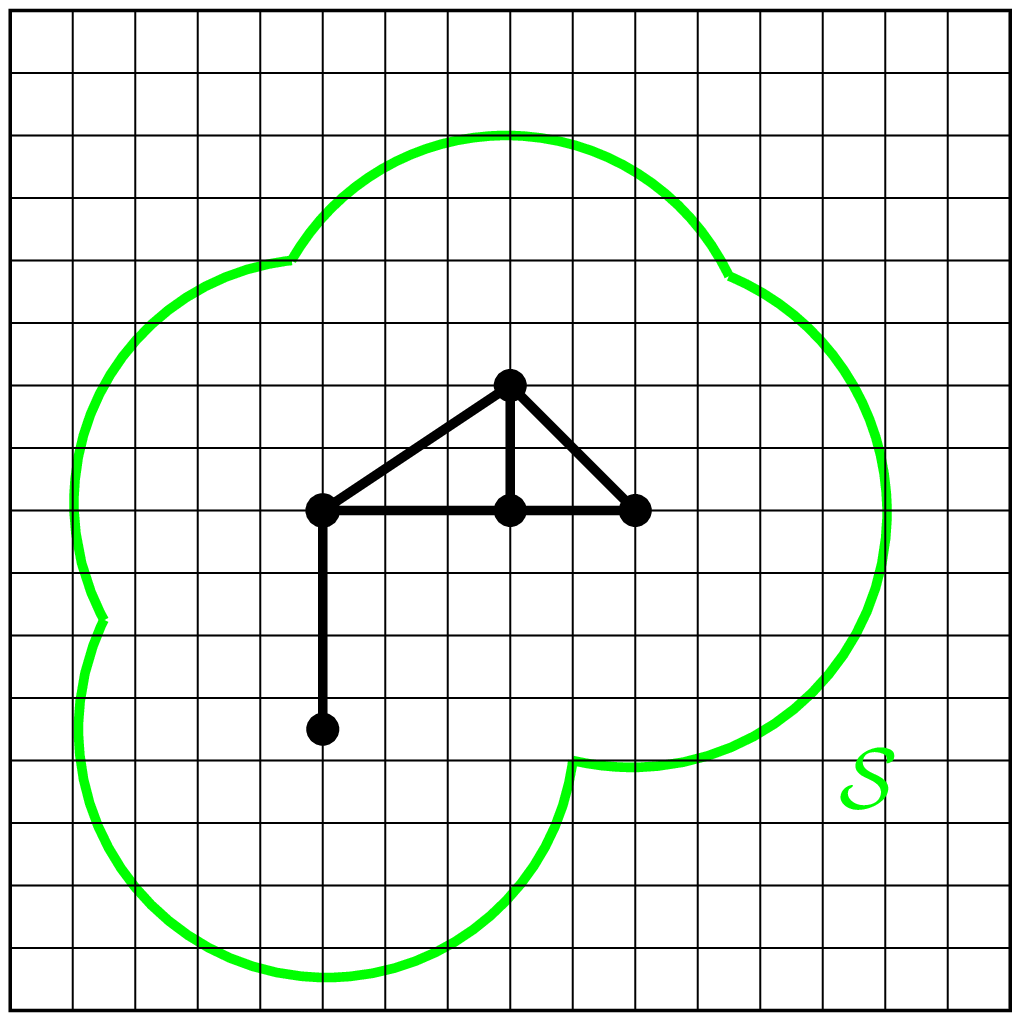,height=3cm,width=3cm}\medspace\medspace\medspace
\medspace\medspace\medspace\epsfig{figure=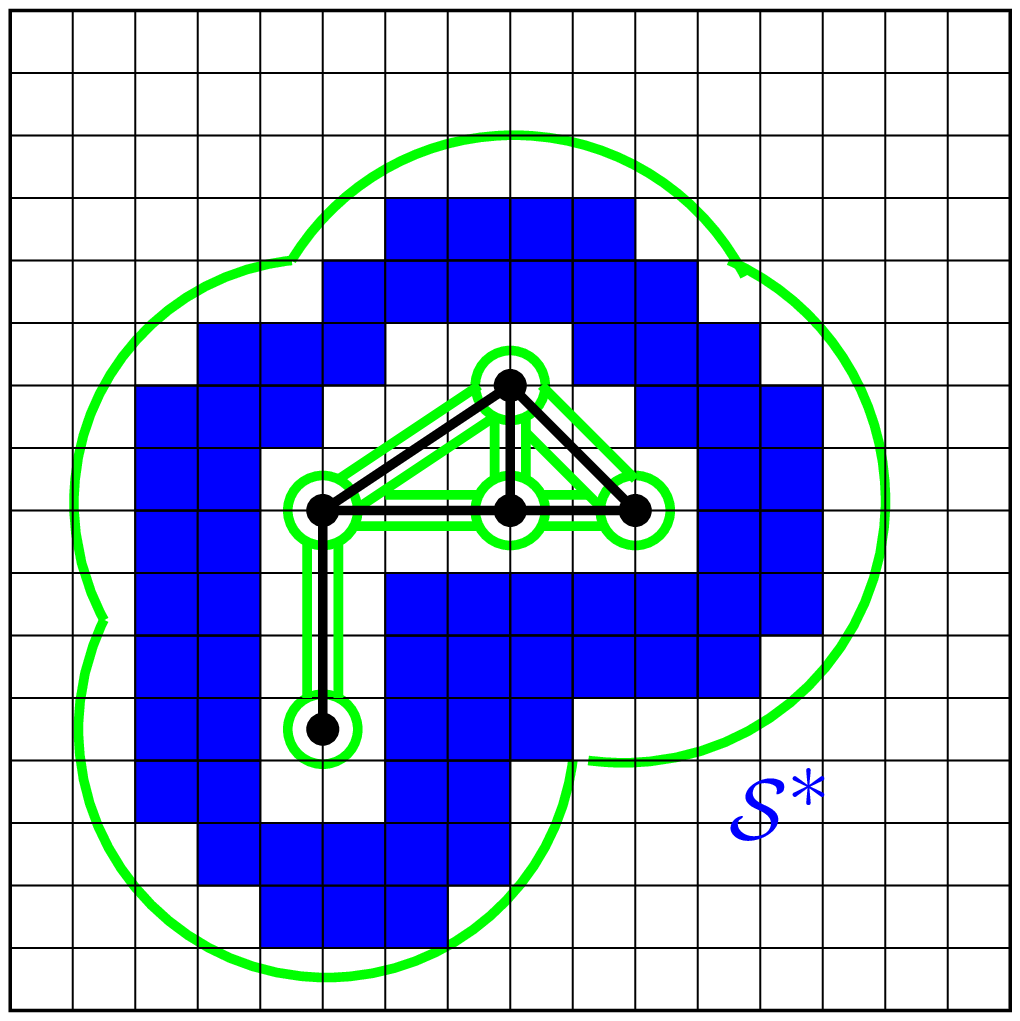,height=3cm,width=3cm}}
\caption{The tessellation for counting components of type~3.}
\label{tas1}
\end{figure}

Now let $i_{\mathsf L}$, $i_{\mathsf R}$, $i_{\mathsf T}$ and $i_{\mathsf B}$ be respectively the indices of the leftmost,
 rightmost, topmost and  bottommost vertices in $\Gamma$ (some of these indices possibly equal).
As in the previous setting, assume that the vertical length of $\Gamma$ (i.e.\ the vertical distance between $X_{i_{\mathsf T}}$ 
and $X_{i_{\mathsf B}}$) is at least $\epsilon r/\sqrt2$. Otherwise, the horizontal length of $\Gamma$ has this property 
and we can rotate the descriptions in the argument. The upper halfcircle with center $X_{i_{\mathsf T}}$ and the lower halfcircle with center $X_{i_{\mathsf B}}$ are disjoint and are contained in $\cS'$. 
If $X_{i_{\mathsf R}}$ is at greater vertical distance from $X_{i_{\mathsf T}}$ than from $X_{i_{\mathsf B}}$, then consider the rectangle of height $\epsilon r/(2\sqrt2)$ and width $r-\epsilon r/(2\sqrt2)$ with one corner on $X_{i_{\mathsf R}}$ and above and to the right of $X_{i_{\mathsf R}}$. Otherwise, consider the same rectangle below and to the right of $X_{i_{\mathsf R}}$.
This rectangle is also contained in $\cS'$ and its interior does not intersect the previously described halfcircles. 
Analogously, we can find another rectangle of height $\epsilon r/(2\sqrt2)$ and width $r-\epsilon r/(2\sqrt2)$ to the left of $X_{i_{\mathsf L}}$ and either above or below $X_{i_{\mathsf L}}$, with the same properties. Hence, taking into account that $\epsilon\le10^{-18}$, we have
\begin{equation}\label{eq:S'bound}
\ar{\cS'}\ge \pi r^2 + 2\left(\frac{\epsilon r}{2\sqrt2}\right)\left(r-\frac{\epsilon r}{2\sqrt2}\right) > \left(1+\frac{\epsilon}{5}\right)\pi r^2.
\end{equation}
Let $\cS^*$ be the union of all the cells in the tessellation which are fully contained in $\cS'$. We loose a 
bit of area compared to $\cS'$. However, if $\alpha$ was chosen small enough, we can guarantee that $\cS^*$ is 
topologically connected and has area $\ar{\cS^*}\ge(1+\epsilon/6)\pi r^2$. This $\alpha$ can be chosen to be the 
same for all components of type~\ref{p:large} (see Figure~\ref{tas1}, right picture).

Hence, we showed that the event $(M_{\ref{p:large}}>0)$ implies that some connected union of cells $\cS^*$ of 
area $\ar{\cS^*}\ge(1+\epsilon/6)\pi r^2$ contains no vertices. By removing some cells from $\cS^*$, we can assume that
$(1+\epsilon/6)\pi r^2 \le \ar{\cS^*} <  (1+\epsilon/6)\pi r^2 + \alpha^2 r^2$.
Let $\cS^*$ be any union of cells with these properties. Note that there are $\Theta(1/r^2)=\Theta(n/\log n)$ 
many possible choices for $\cS^*$. The probability that $\cS^*$ contains no vertices is
\[
(1-\ar{\cS^*})^n \le e^{-(1+\epsilon/6)\pi r^2 n} = \left(\frac\mu n\right)^{1+\epsilon/6}.
\]
Therefore, we can take the union bound over all the $\Theta(n/\log n)$ possible $\cS^*$,
and obtain an upper bound of the probability that there is some component of the type~\ref{p:large}:
\[
\pr(M_{\ref{p:large}}>0) \le \Theta\left(\frac{n}{\log n}\right) \left(\frac\mu n\right)^{1+\epsilon/6}
= \Theta\left(\frac{1}{n^{\epsilon/6}\log n}\right) .
\]
\end{pt}
\begin{pt}\label{p:notembed}
Consider all the possible components in $\nRG$ which are not embeddable and not solitary either.  Call them components of type~\ref{p:notembed}, and let $M_{\ref{p:notembed}}$ denote their number.

We tessellate the torus $[0,1)^2$ into $\Theta(n/\log n)$ small square cells of side length $\alpha r$, where $\alpha>0$ is a sufficiently small positive constant.

Let $\Gamma$ be a component of type~\ref{p:notembed}.
Let $\cS=\cS_\Gamma$ be the set of all points in the torus $[0,1)^2$ which are at distance at most $r$ from some vertex in $\Gamma$.
Remove from $\cS$ the vertices of $\Gamma$ and the edges (represented by straight segments) 
and denote by $\cS'$ the remaining set. By construction, $\cS'$ must contain no vertex in $\cX$.

Suppose there is a horizontal or a vertical band of width $2r$ in $[0,1)^2$ which does not intersect the component $\Gamma$ (assume w.l.o.g. that it is the topmost horizontal band consisting of all points with the $y$-coordinate in $[1-2r,1)$).
Let us divide the torus into vertical bands of width $2r$. All of them must contain at least one vertex of $\Gamma$, since otherwise $\Gamma$ would be embeddable. Select any $9$ consecutive vertical bands and pick one vertex of $\Gamma$ with maximal $y$-coordinate in each one. For each one of these $9$ vertices, we select the left upper quartercircle centered at the vertex if the vertex is closer to the right side of the band or the right upper quartercircle otherwise. These nine quartercircles we chose are disjoint and must contain no vertices by construction. Moreover, they belong to the same connected component of the set $\cS'$, which we denote by $\cS''$, and which has an area of $\ar{\cS''}\ge(9/4)\pi r^2$.
Let $\cS^*$ be the union of all the cells in the tessellation of the torus which are completely contained in $\cS''$.
We lose a bit of area compared to $\cS''$. However,
as usual, by choosing $\alpha$ small enough we can guarantee that $\cS^*$ is connected and it has an area of $\ar{\cS^*}\ge(11/5)\pi r^2$.
Note that this $\alpha$ can be chosen to be the same for all components $\Gamma$ of this kind.
\remove{\begin{figure}
\centerline{\epsfig{figure=tasellation6.eps,height=3cm,width=3cm}\medspace\medspace\medspace
\medspace\medspace\medspace\epsfig{figure=tasellation7.eps,height=3cm,width=3cm}} 
\caption{The tassellation for counting components of type~4.}
\label{tas2}
\end{figure}}

Suppose otherwise that all horizontal and vertical bands of width $2r$ in $[0,1)^2$ contain at least one 
vertex of $\Gamma$. Since $\Gamma$ is not solitary it must be possible that it coexists with some other 
non-embeddable component $\Gamma'$. Then all vertical bands or all horizontal bands of 
width $2r$ must also contain some vertex of $\Gamma'$ (assume w.l.o.g. the vertical bands do). Let us 
divide the torus into vertical bands of width $2r$. We can find a simple path $\Pi$ with vertices in $\Gamma'$ which passes 
through $11$ consecutive bands. For each one of the $9$ internal bands, pick the uppermost vertex of 
$\Gamma$ in the band below $\Pi$ (in the torus sense). As before each one of these vertices 
contributes with a disjoint quartercircle which must be empty of vertices, and by the same argument 
we obtain a connected union of cells of the tessellation, which we denote by $\cS^*$, with $\ar{\cS^*}\ge(11/5)\pi r^2$ and containing no vertices.

Hence, we showed that the event $(M_{\ref{p:notembed}}>0)$ implies that some connected union of cells $\cS^*$ with $\ar{\cS^*}\ge(11/5)\pi r^2$ contains no vertices.
By repeating the same argument we used for components of type~\ref{p:large} but replacing $(1+\epsilon/6)\pi r^2$ by $(11/5)\pi r^2$, we get
\[
\pr(M_{\ref{p:notembed}}>0) = \Theta\left(\frac{1}{n^{6/5}\log n}\right) .
\]
\end{pt}
\end{proof}
For a random variable $X$ and any $k \geq 1$, we denote by $\ex[X]_k$ the $k$th factorial moment of $X$, i.e. $\ex [X]_k=\ex[X(X-1)\ldots(X-k+1)]$.
\begin{lem}\label{lem:EZei2}
Let  $\ell\ge2$ be a fixed integer. Let $0<\epsilon<1/2$ be fixed. 
Assume that $\mu=\Theta(1)$. Then
\[
\EXP{K'_{\epsilon,\ell}}_2 = O(1/\log^{2\ell-2} n).
\]
\end{lem}
\begin{proof}
As in the proof of Lemma~\ref{lem:EZei}, we assume that each component $\Gamma$ which contributes to $K'_{\epsilon,\ell}$ has a unique leftmost vertex $X_i$, and the vertex $X_j$ in $\Gamma$ at greatest distance from $X_i$ is also unique. In fact, this happens with probability $1$.

Choose any two disjoint subsets of $\nset$ of size $\ell$ each, namely $J_1$ and $J_2$, with four distinguished elements $i_1,j_1\in J_1$ and $i_2,j_2\in J_2$.
For $k\in\{1,2\}$, denote by $\cY_k = \bigcup_{l\in J_k} X_l$ the set of random points in $\cX$ with indices in $J_k$.
Let $\cE$ be the event that the following conditions hold for both $k=1$ and $k=2$:
All vertices in $\cY_k$ are at distance at most $\epsilon r$ from $X_{i_k}$ and to the right of $X_{i_k}$; vertex $X_{j_k}$ is the one in $\cY_k$ with greatest distance from $X_{i_k}$; and the vertices of $\cY_k$ form a component $\Gamma$ of $\nRG$.
If $\pr(\cE)$ is multiplied by the number of possible choices of $i_k$, $j_k$ and the remaining vertices of $J_k$, we get
\begin{equation}\label{eq:EZei2}
\ex[K'_{\epsilon,\ell}]_2 = O(n^{2\ell}) \pr(\cE).
\end{equation}

In order to bound the probability of $\cE$ we need some definitions.
For each $k\in\{1,2\}$, let $\rho_k=d(X_{i_k},X_{j_k})$ and let $\cS_k$ be the set of all the points in the torus $[0,1)^2$ which are at distance at most $r$ from some vertex in $\cY_k$. Obviously $\rho_k$ and $\cS_k$ depend on the set of random points $\cY_k$. Also define $\cS=\cS_1\cup\cS_2$.

Let $\cF$ be the event that $d(X_{i_1},X_{i_2})>3r$. This holds with probability $1-O(r^2)$. In order  to bound $\pr(\cE\mid\cF)$, we apply a similar approach to the one in the proof of Lemma~\ref{lem:EZei}. In fact,  observe that if $\cF$ holds then $\cS_1\cap\cS_2=\emptyset$.
Therefore in view of~\eqref{eq:Sbound} we can write
\begin{equation}	\label{eq:Sbound2}
\pi r^2(2+(\rho_1+\rho_2)/(6r))<\ar{\cS}<\frac{18\pi}{4}r^2,
\end{equation}
and using the same techniques that gave us~\eqref{eq:Pbound} we get
\begin{equation}\label{eq:Pbound2}
(1-\ar{\cS})^{n-2\ell} <  \left(\frac\mu n\right)^{2+(\rho_1+\rho_2)/(6r)} \frac1{(1-18\pi r^2/4)^{2\ell}}.
\end{equation}
Observe that $\cE$ can also be described as follows: For each $k\in\{1,2\}$ there is some non-negative real $\rho_k\le\epsilon r$ such that $X_{j_k}$ is placed at distance $\rho_k$ from $X_{i_k}$ and to the right of $X_{i_k}$; all the remaining vertices in $\cY_k$ are inside the halfcircle of center $X_{i_k}$ and radius $\rho_k$;
and the $n-\ell$ vertices not in $\cY_k$ lie outside $\cS_k$.
In fact, rather than this last condition, we only require for our bound that all vertices in $\cX\setminus (\cY_1\cup \cY_2)$ are placed outside $\cS$, which has probability $(1-\ar{\cS})^{n-2\ell}$. Then, from~\eqref{eq:Pbound2} and  following an analogous argument to the one that leads to~\eqref{eq:PE},
we obtain the bound
\begin{align*}
\pr(\cE\mid\cF) &\le \Theta(1) \int_0^{\epsilon r}\int_0^{\epsilon r} \pi\rho_1 \left(\frac\pi2 \rho_1^2\right)^{\ell-2} \pi\rho_2 \left(\frac\pi2 \rho_2^2\right)^{\ell-2}
\frac1{n^{2+(\rho_1+\rho_2)/(6r)}} \, d\rho_1d\rho_2
\\
&= \Theta(1) \: I(1/6)^2,
\end{align*}
where $I(1/6)$ is defined in~\eqref{eq:Ibeta}. Thus from~\eqref{eq:Ibeta2} we conclude
\begin{equation}\label{eq:PEF}
\pr(\cE\wedge\cF) \le \Theta(1) \: P(\cF) \: I(1/6)^2 = O\left(\frac{1}{n^{2\ell} \log^{2\ell-2} n}\right).
\end{equation}

Otherwise, suppose that $\cF$ does not hold (i.e.\ $d(X_{i_1},X_{i_2})\le3r$). Observe that $\cE$ implies that $d(X_{i_1},X_{i_2})>r$, since $X_{i_1}$ and $X_{i_2}$ must belong to different components. Hence the circles with  centers on $X_{i_1}$ and $X_{i_2}$ and radius $r$ have an intersection of area less than $(\pi/2)r^2$. These two circles are contained in $\cS$
and then we can write $\ar{\cS}\ge(3/2)\pi r^2$.
Note that $\cE$ implies that all vertices in $\cX\setminus (\cY_1\cup \cY_2)$ are placed outside $\cS$ and that for each $k\in\{1,2\}$ all the vertices in $\cY_k\setminus\{X_{i_k}\}$ are at distance at most $\epsilon r$ and to the right of $X_{i_k}$.
This gives us the following rough bound
\[
\pr(\cE\mid\overline\cF) \le \left(\frac\pi2(\epsilon r)^2\right)^{2\ell-2}  \left(1-\frac{3\pi}2 r^2\right)^{n-2\ell}
= O(1) \left(\frac{\log n}{n}\right)^{2\ell-2}  \left(\frac\mu n\right)^{3/2}.
\]
Multiplying this  by $\pr(\overline\cF)=O(r^2)=O(\log n/n)$ we obtain
\begin{equation}\label{eq:PEnF}
\pr(\cE\wedge\overline\cF) = O\left(\frac{\log^{2\ell-1} n}{ n^{2\ell+1/2}}\right),
\end{equation}
which is negligible compared to \eqref{eq:PEF}. The statement follows from~\eqref{eq:EZei2}, \eqref{eq:PEF} and~\eqref{eq:PEnF}.
\end{proof}

Our main theorem now follows easily:
{From} Corollary~1.12 in~\cite{Bollobas01}, we have
\[
\ex K'_{\epsilon,\ell} - \frac12\ex [K'_{\epsilon,\ell}]_2 \le \pr(K'_{\epsilon,\ell}>0) \le \ex K'_{\epsilon,\ell},
\]
and therefore by Lemmata~\ref{lem:EZei} and~\ref{lem:EZei2} we obtain
\[
\pr(K'_{\epsilon,\ell}>0) = \Theta(1/\log^{\ell-1} n).
\]
Combining this and Lemma~\ref{lem:PY}, yields the statement.

\section{Proof of Corollary~\ref{cor:hitting}}\label{sec:corollary}

Before proving Corollary~\ref{cor:hitting}, we give a proof of Proposition~\ref{prop:wellknown}, since we will make use of the arguments used in the proof of this proposition.
\begin{proof}[Proof of Proposition~\ref{prop:wellknown}.]
Recall that $\mu=n e^{-\pi r^2 n}$ and $r=\sqrt{\frac{\log n-\log\mu}{\pi n}}$. Observe that $r\in[0,+\infty)$ is monotonically decreasing with respect to $\mu\in(0,n]$. Hence, the probability that $\nRG$ is connected is also decreasing with respect to $\mu$.

Suppose first that $\mu=\Theta(1)$. From~\eqref{eq:EK1} and since $O(r^4n)=o(1)$ we have that $\ex K_1\sim\mu$.
We shall compute the factorial moments of $K_1$ and show that $\ex[K_1]_k\sim\mu^k$ for each fixed $k$.
As in Lemma~\ref{lem:EZei}, for $k \geq 2$, we fix an arbitrary set of indices $J\subset\nset$ of size $|J|=k$. Denote by $\cY = \bigcup_{k\in J} X_k$ the set of random points in $\cX$ with indices in $J$.
Let $\cE$ be the event that all vertices in $\cY$ are isolated, and denote by $\cS$ the set of points in $[0,1)^2$ that are at distance at most $r$ from some vertex in $\cY$. We have $\ex[K_1]_k \sim n^k \pr(\cE)$. Note that in order for the event $\cE$ to happen, we must have $\cS \cap (\cX \setminus \cY) = \emptyset$.
To compute $\pr(\cE)$, we distinguish two cases:
\par\noindent
\emph{Case 1:} Suppose that $\forall i \neq j \in J$, $d(X_i,X_j) > 4r$. In this case, $\ar{\cS} =kr^2 \pi$, and thus the probability of $\cE$ is  $(1-kr^2 \pi)^{n-k} \sim e^{-kr^2 \pi n}$.
\par\noindent
\emph{Case 2:} Otherwise there exists $\exists i \neq j \in J$ such that $d(X_i,X_j) \leq 4r$. Define $J' =\{j \in J \mid \exists i \in J, i < j,\,d(X_i, X_j) \leq 4r\}$ and let $\ell=|J'|$. Note that $1\le \ell \le k-1$. Let $j'$ be the smallest element of $J'$ and let $i' < j'$ be the (smallest) element of $J$ with $d(X_{i'},X_{j'}) \leq  4r$. Denote by $C_{i'}$ the circle of radius $r$ centered at $X_{i'}$, and consider the halfcircle of radius $r$ centered at $X_{j'}$ delimited by the line going through $X_{j'}$, perpendicular to the line connecting $X_{i'}$ with $X_{j'}$, and which does not intersect $C_{i'}$ (note that $d(X_{i'},X_{j'}) > r$, so this halfcircle exists). This circle and halfcircle contribute to $\ar{\cS}$ by $\frac32 r^2 \pi$, and thus in total $\ar{\cS} \geq (k-\ell+\frac12)r^2 \pi$.
Moreover, the probability that any $j \in J$ to belongs to $J'$ is at most $\Theta(r^2)$. Hence, if we denote by $\mathcal J_\ell$ the event that such a set $J'$ with $|J'| = \ell$ exists, we have for any $1\le\ell\le k-1$,
\[ 
\pr(\cE \mid \mathcal J_\ell)\pr(\mathcal J_\ell)
\leq (1-(k-\ell+1/2))^{r^2 \pi n}\Theta(r^2)^\ell = o(e^{-kr^2 \pi n}).
\]
Then, the main contribution to $\pr(\cE)$ comes from Case 1, and therefore $\ex K_1\sim n^k e^{-kr^2 \pi n} = \mu^k$, so the random variable $K_1$ is asymptotically Poisson with parameter $\mu$. By Theorem~\ref{thm:static2}, a.a.s.\ $\nRG$ consists only of isolated vertices and a solitary component, and the second statement in the result is proven.

The first and third statements follow directly from the fact that, for any $\mu=\Theta(1)$, $\pr(\nRG\text{ is connected})\sim e^{-\mu}$,  combined with the decreasing monotonicity of this probability with respect to $\mu$.
\end{proof}

\begin{proof}[Proof of Corollary~\ref{cor:hitting}.]
For any $\epsilon > 0$, one can find a large enough constant $\kappa=\kappa(\epsilon)$ such that $e^{-e^{\kappa}} < \epsilon/2$ and $1-e^{-e^{-\kappa}}< \epsilon/2$. Let $r_\ell=\sqrt{\frac{\log n-\kappa}{\pi n}}$ and $r_u=\sqrt{\frac{\log n+\kappa}{\pi n}}$.
By Proposition~\ref{prop:wellknown}, $K_1$ is asymptotically Poisson in $\nRGl$ and $\nRGu$, with parameter $\mu=e^{\kappa}$ and $\mu=e^{-\kappa}$ respectively. Therefore, in $\nRGl$ we have $\pr(K_1=0)\sim e^{-e^\kappa}<\epsilon/2$, and in $\nRGu$ we have $\pr(K_1>0)\sim 1-e^{-e^{-\kappa}}<\epsilon/2$.
Moreover, by Theorem~\ref{thm:static2}, a.a.s.\  both $\nRGl$ and $\nRGu$ consist only of isolated vertices and a giant solitary component.
Hence, with probability at least $1-\epsilon$, the random process $\nRGR$ has the following evolution: for $r\le r_\ell$, the graph stays disconnected; at $r=r_\ell$, there are only a few isolated vertices and a giant component; for $r$ between $r_\ell$ and $r_u$, all isolated vertices merge together or with other components; finally for $r\ge r_u$, the graph is connected.
For this particular evolution of the process, $r_c=r_i$ unless for an $r$ with $r_\ell < r < r_u$ some isolated vertices merge together and create a small component before being absorbed by the giant one. Then, it is sufficient for our purposes to show that a.a.s.\ any two isolated vertices in $\nRGl$ are at a distance bigger than $r_u$.

Define $Z$ to be the random variable that counts the pairs of vertices $i$ and $j$ which are both isolated in $\nRGl$ and such that $d(X_i,X_j) \leq r_u$. By the same argument as in the proof of Proposition~\ref{prop:wellknown}, setting $\cS$ to be the set of points in $[0,1)^2$ at distance at most $r_\ell$ from either $X_i$ or $X_j$, we obtain $\ar{\cS} \geq \frac32 r_\ell^2 \pi$. Moreover, since $r_\ell < d(X_i,X_j) \leq r_u$, $X_j$ must lie in an annulus of area $\Theta(1/n)$ around $X_i$, which occurs with probability $\Theta(1/n)$. Taking a union bound over all pairs of vertices $i$ and $j$,
\[
\pr(Z>0) \leq n(n-1) \left(1-\frac{3}{2}r_\ell^2\pi\right)^{n-2} \; \Theta(1/n) = \Theta\left(n^{-1/2}\right).
\]
Therefore, when gradually increasing $r$ from $r_\ell$ to $r_u$, a.a.s.\ no pair of isolated vertices in $\nRGl$ gets connected before joining the solitary component, and thus no component of size $2$ or larger (except for the solitary component) appears in this part of the process. Hence, with probability at least $1-\epsilon$, we have that $r_c=r_i$, and the statement follows, since $\epsilon$ can be chosen to be arbitrarily small.
\end{proof}
\paragraph{Acknowledgment.} We thank an anonymous referee for suggesting the application of Theorem~\ref{thm:static2} to obtain Corollary~\ref{cor:hitting}.
\bibliographystyle{alpha}

\end{document}